\documentclass{article}
\usepackage[T1]{fontenc}                
\usepackage[utf8]{inputenc}             
\usepackage{amssymb,amsthm,amsmath}    
\usepackage{enumerate}
\usepackage{paralist}
\usepackage{authblk}
\usepackage{graphicx}
\usepackage{enumerate}
\usepackage{array}
\newcolumntype{C}[1]{>{\centering\let\newline\\\arraybackslash\hspace{0pt}}m{#1}}
\newcolumntype{P}[1]{>{\centering\arraybackslash}p{#1}}
\usepackage{tikz}
\usetikzlibrary{calc}
\usetikzlibrary{decorations.pathreplacing}

\newtheorem{theorem}{Theorem}
\newtheorem{lemma}{Lemma}
\newtheorem{prob}{Problem}
\newtheorem*{prob*}{Problem}

\newtheorem*{prop*}{Proposition}
\newtheorem*{example}{Example}

\newtheorem*{observ*}{Observation}

\newcommand{\Z}{\mathbb Z}

\newcommand{\M}{\mathcal M}

\newcommand{\ov}[1]{\overline{#1}}

\newcommand{\derive}[1]{\xrightarrow{#1}}

\newcommand{\om}{\omega}
\newcommand{\Si}{\mathsf\Sigma}

\newcommand{\Ga}{\mathsf\Gamma}

\newcommand{\fa}{\forall}

\title{On Bi-infinite and Conjugate Post Correspondence Problems}
\author{Olivier Finkel} 
\affil[1]{Institut de Math\'ematiques de Jussieu - Paris Rive Gauche\\
CNRS, Universit\'e Paris Cit\'e,  Sorbonne Universit\'e, Paris, France. \textit{finkel@math.univ-paris-diderot.fr}}
\author{Vesa Halava\thanks{Supported by emmy.network foundation under the aegis of the Fondation de Luxembourg.}}
\author{Tero Harju} 
\author{Esa Sahla}
\affil[2]{Department of Mathematics and Statistics, University of Turku, Finland. \textit{\{vesa.halava,harju,esa.sahla\}@utu.fi}}

\newcommand{\tikzmark}[1]{\tikz[overlay,remember picture] \node (#1) {};}

\newcommand*{\BraceAmplitude}{0.5em}
\newcommand*{\VerticalOffset}{1ex}
\newcommand*{\InsertUnderBrace}[4][]{%
    \begin{tikzpicture}[overlay,remember picture]
\draw [decoration={brace,amplitude=\BraceAmplitude},decorate, thick,draw=blue,text=black,#1]
        ($(#3)+(0,-\VerticalOffset)$) -- 
        ($(#2)+(0,-\VerticalOffset)$)
        node [below=\VerticalOffset, midway] {#4};
    \end{tikzpicture}%
}%
\begin{document}

\maketitle

\begin{abstract}

We study two modifications of the Post Correspondence Problem (PCP), namely 1) the bi-infinite version, where  it is asked whether there exists a bi-infinite word such that two given morphisms agree on it, and     2) the conjugate version, where we require the images of a solution for two given morphisms are conjugates of each other. 
For the bi-infinite PCP we show that it is in the class  $\Sigma_2^0$ of the arithmetical hierarchy and for the 
conjugate PCP we give an undecidability proof by reducing it to the word problem for a special type of semi-Thue systems.

\vspace{1em}

\end{abstract}

\section{Introduction}
The original formulation of the  \emph{Post Correspondence Problem} (PCP) by Emil Post in \cite{Po-46} is the following:

\begin{prob}[PCP]
Let $A$ be a finite alphabet. Given a finite set  of pairs of words over $A$, say $(u_1,v_1),(u_2,v_2),  \dots,$ $(u_n,v_n)$,  does there exists a nonempty sequence $i_1,\dots,i_k$ of indices such that 
\[
u_{i_1}u_{i_2}\cdots u_{i_k} =v_{i_1}v_{i_2}\cdots v_{i_k}\, ?
\]
\end{prob}

Post proved that the PCP is undecidable in \cite{Po-46}. Since then  
the PCP and its many variants have been used as a bridge from combinatorial undecidable problems of computational systems and formal rewriting systems to decision problems in algebraic settings. 
The PCP is usually defined as a problem in free word monoids (as $A^*$ is the free monoid of all  finite words over $A$ with catenation as the operation). Indeed, the PCP is equivalent to asking for two given morphisms $g,h\colon B^*\to A^*$, whether or not there exists a non-empty word $w$ such that 
$$
g(w)=h(w).
$$
Note that we may choose  set $B=\{1,\dots, n\}$, and $g(i)=u_i$, $h(i)=v_i$ for all $i=1,\dots, n$ where $(u_i,v_i)$ is a pair  in the original formulation of the PCP. 
The variants of the  PCP also reveal the boundary between decidability and undecidability. It is known that the PCP is decidable for $n=2$, see \cite{EKR82,HHH02}, and undecidable for $n=5$, see \cite{Neary}. On the other hand, it is known that the infinite PCP, asking whether there is a (right) infinite sequence $i_1,i_2,\dots$ of the indices such that the words agree, is decidable for two pairs of words, see \cite{HHK}, and undecidable for 8 pairs, see \cite{dongliu}. It has  been proved that the infinite PCP is not "more complex" than the PCP with respect to the arithmetical hierarchy, see \cite{Finkel}, where it was proved that the infinite PCP is $\Pi_1^0$-complete  as the PCP is known to be $\Sigma_1^0$-complete.   
 
In this paper we study two variants of the PCP. The first variant is called the 
\emph{bi-infinite Post Correspondence Problem} ($\Z$PCP), where it is asked whether or not there exits a bi-infinite sequence of the indices such that the words agree.  The morphisms version of the  $\Z$PCP is the following: 

\begin{prob}\label{prob:bi-inf1}
Given two morphisms $h,g \colon A^*\to B^*$, does the exist a  bi-infinite word $w$ such that $h(w) = g(w)$. 
\end{prob}

Note that already the equality of the images of bi-infinite words needs to be defined properly: for a bi-infinite word $w$, $h(w) = g(w)$ if and only if there is a constant $s\in \Z$ such that for all letters $h(w)(i)=g(w)(i+s)$ for all positions $i \in \Z$. 
An \emph{instance} of the $\Z$ PCP is a pair of morphisms $(h,g)$ and a bi-infinite word $w$ satisfying  $h(w)=g(w)$ is said to be a \emph{solution} of the instance $(h,g)$. 

Our second variant deals with conjugate words.
Two words $x$ and $y$ are \emph{conjugates} if there exist words $u$ and $v$ such that  $x=uv$ and $y=vu$. 

We call the following problem the \emph{conjugate-PCP}. 

\begin{prob}
Given two morphisms $h,g: A^* \rightarrow B^*$, does there exist a word $w \in A^+$ such that $h(w)=uv$ and $g(w)=vu$ for some words $u,v \in B^*$.
\end{prob}

The behaviour of the instances of the conjugate-PCP differ vastly from the more traditional variants where a valid presolution (prefix of a candidate solution) can be verified by aligning the
matching parts of the images. Working with the possible solutions of the
instances of the conjugate-PCP is less intuitive.

For example let us have morphisms $h,g$ and we guess that a solution $w$ begins with the letter $a$. Then the situation is the following:

\begin{alignat*}{2}
 h(w) \ = \ &\tikzmark{Start1} h(a) \qquad \cdots \qquad \tikzmark{End1} \tikzmark{Start2} g(a) \qquad \qquad \cdots \qquad \qquad \tikzmark{End2} \\
 \\
 g(w) \ = \ & \tikzmark{Start3} g(a) \qquad \qquad \cdots \qquad \qquad  \tikzmark{End3}\tikzmark{Start4} h(a) \qquad \cdots \qquad\tikzmark{End4} \\
\end{alignat*}
\InsertUnderBrace[draw=black,text=black]{Start1}{End1}{$u$}
\InsertUnderBrace[draw=black,text=black]{Start2}{End2}{$v$}
\InsertUnderBrace[draw=black,text=black]{Start3}{End3}{$v$}
\InsertUnderBrace[draw=black,text=black]{Start4}{End4}{$u$}

The validity of the presolution $a$ cannot be verified because  there may  not be any matching between $h(a)$ and $g(a)$. Moreover the factorization of the images to $u$ and $v$ need not be unique even for minimal solutions:

\begin{example}
\normalfont
Let $h,g: \lbrace a,b \rbrace^* \rightarrow \lbrace a,b \rbrace^*$ be morphisms defined by
\begin{alignat*}{4}
h(a) &= aba, \ \ &&g(a) &&= bab, \\
h(b) &= b, &&g(b) &&= a.
\end{alignat*}
Now $ab$ is a minimal solution for the conjugate-PCP instance $(h,g)$ having two factorizations: $u=a, v=bab$ or $u= aba, v=b$.
\end{example}

Both variants  defined above were originally proved to be undecidable in \cite{Ruo} using linearly bounded automata (LBA)\footnote{Note that in \cite{Ruo} 
the $\Z$PCP is called \emph{doubly infinite PCP}}. Indeed, the conjugate-PCP was not directly proved in \cite{Ruo}, although it is claimed so in \cite{Ruo2}. Let us consider the terminology of Ruohonen in  \cite{Ruo2} in a bit more details: Let $u$ and $v$ be words, and denote $u\sim_m v$ if there exist words $u_1,\dots, u_m$ and a permutation $\tau$ on the set $\{1,\dots, m\}$   such that $u=u_1\cdots u_m$ and $v=u_{\tau(1)}\cdots u_{\tau(m)}$. 

\begin{prob}[$(m,n)$-permutation PCP] 
Given morphisms $g,h\colon A^*\to B^*$ does there exists words $u,v\in A^*$ such that
$$
u\sim_m v \text{ and } g(u)\sim_n h(v).
$$
\end{prob}

Obviously, our formulation of the conjugate-PCP is the (1,2)-permutational PCP of Ruohonen. Now in 
\cite{Ruo2}, it is mentioned that (1,2)-permutational PCP was proved to be undecidable in \cite{Ruo}, but  the problem is not explicitly mentioned there. On the other hand, the (2,2)-permutational PCP is shown to be undecidable in \cite{Ruo}\footnote{Note that $(2,2)$-permutational PCP is called the \emph{PCP for circular words} in \cite{Ruo}}. It is possible that Ruohonen uses the later result without details, because of the following simple lemma, which follows from special cyclic shift property of permutations.

\begin{lemma}\label{simplest}
For morphisms $g,h : \colon A^*\to B^*$, the instance $(g,h)$ has a solution to the  $(1,2)$-permutational PCP if and only if it has a solution to the $(2,2)$-permutational PCP.
\end{lemma}

\begin{proof}
Firstly, a solution to the $(1,2)$-permutational PCP is a solution to the $(2,2)$-permutational PCP where the first permutation on the pre-image being trivial. 

Secondly, assume that there exists a solution $u=xy,v=yx$, two $(2,2)$-permutational PCP. So  $g(xy)=zw$ and $h(yx)=wz$ for some  $w,z\in B^*$. We have two cases, either $h(y)$ is a prefix of $w$, or vice versa. 

In the first case, $w=h(y)r$ for some word $r\in B^*$, and $h(xy)=rzh(y)$ and $g(xy)=zw=zh(y)r$ and, therefore, $xy$ is a solution for the $(1,2)$-permutational PCP. 

In the second case, $h(y)=wr$, $z=rh(x)$, for some word $r\in B^*$. Then  $g(xy)=rh(x)w$ and  $h(xy)=h(x)wr$,  implying that $xy$ is again a solution for the $(1,2)$-permutational PCP. 
\end{proof}

We stress that the undecidability of the $(m,n)$-permutational PCP was proved in \cite{Ruo2}, using the machinery of LBA's used already in \cite{Ruo}, for all $m$ and $n$. This extends the result is \cite{Ruo} where it was shown that the $(n,1)$-permutational PCP is undecidable for all $n$\footnote{Note that in \cite{Ruo} $(n.1)$-permutational PCP is called $n$-permutational PCP.}. 

The proof in \cite{Ruo} and \cite{Ruo2} are rather involved because of the employment of the computations of LBA's, and there is a quest for simpler treatments of the problems. There is a line of new simplified proofs for Ruohonen's results on the permutational PCP's, which  use of a word problem of the special type of semi-Thue systems:
For the case of $(2,1)$-permutational PCP  a somewhat simpler proof was given in  \cite{HaHa-13}, where the problem was called  the \emph{circular PCP}.  
In \cite{Ernvall2015} a simplified proof for $(n,1)$-permutational PCP (or $n$-permutational PCP) was given for all positive $n$. For $\Z$PCP a simpler proof was given in \cite{bi-inf}. 

In the next section we give a new proof to  undecidability of the conjugate-PCP by reducing it to the word problem for a special type of semi-Thue systems. These are is indeed the same special semi-Thue systems that were used for provingundecidability of the $\Z$PCP in \cite{bi-inf}, but the construction here is different due to differences in the $\Z$PCP and the conjugate-PCP.
By Lemma~\ref{simplest}, this also proves undecidability of the $(2,2)$-permutational PCP. 

In the final section we show that the $\Z$PCP is in  the class $\Sigma_2^0$ of the arithmetical hierarchy. This reflects to the result in \cite{Finkel}, where it was proved that the infinite PCP is $\Pi_1^0$-complete  as the PCP is known to be $\Sigma_1^0$-complete.   

At first sight, it may seem that the $\Z$PCP and the conjugate PCP do not have anything in common, but that is not the case. In both of these problems, solutions have a shift, in the $\Z$PCP the shift makes the images, two bi-infinite words, equal and in the conjugate PCP the images are equal over a cyclic shift of the word. Therefore, the construction in the next section for the conjugate PCP has similar ideas as the construction for the $\Z$PCP in \cite{bi-inf}.    

\section{The proof of undecidability of the conjugate-PCP}

We shall shortly recall the construction of the semi-Thue system $T_\M$ in \cite{bi-inf}. First of all,  
a \emph{semi-Thue system}
$T$ is a pair  $(\Gamma,R)$ where  $\Gamma= \{ a_1, a_2, \dots , a_n\}$ is a finite alphabet,
the elements of which are called \emph{generators} of~$T$, and the relation $R \subseteq \Gamma^*
\times \Gamma$ is  the set of
\emph{rules} of $T$. We write $u \derive{}_{T} v$, if there exists a rule $(x,y) \in R$ such that $u=u_1xu_2$ and
$v=u_1yu_2$ for some words $u_1,u_2\in \Gamma^*$. We denote by $\derive{}^*_{T}$ the
reflexive and transitive closure of $\derive{}_{T} $, and by $\derive{}^+_{T}$ the transitive
closure of $\derive{}_{T} $. Note that the index $T$ is omitted from the notation, when the semi-Thue system studied is clear from the
context. If $u\derive{}^* v$ in $T$, we say that there is a derivation from  
$u$ to $v$ in $T$. 

In the \emph{word problem} for semi-Thue systems, is it asked, for a given semi-Thue system $T$ and words $w$ and $u$, whether $w\derive{}_T^* u$. In a \emph{circular word problem} on the other hand, it is asked whether there exists a word $u$ for a given semi-Thue system $T$ such that $u\derive{}_T^+ u$.    

In \cite{bi-inf}, a special kind of semi-Thue system was constructed which harnesses the structure of a given deterministic Turing machine. Assume that a \emph{Turing 
machine}, TM for short,  $\mathcal{M}$ is of the form $\mathcal{M}=(Q, \Si, \Ga, \delta, q_0, F)$, where $Q$ is a finite set of states, $\Si$ is a finite input alphabet, $\Ga$ is a finite tape alphabet satisfying $\Si  \subseteq \Ga$, containing a special blank symbol $\Box \in \Ga \setminus \Si$, 
$q_0$ is a unique \emph{initial state},  $\delta$ is a \emph{transition} mapping from $Q \times \Ga$ to subsets of $Q \times \Ga \times \{L, R, S\}$, and $F\subseteq Q$ is the set of \emph{accepting states}. We assume that in a TM, the transition mapping is a partial function as the TM's are assumed to be deterministic. 

For purposes of this section we also assume that $F=\{H\}$, that is, there exists a unique accepting state $H\in Q$, called  the \emph{halting} state. It can be assumed that a computation of a TM halts (i.e., no more transitions are applicable) if and only it arrives to state $H$. We denote a configuration of $\M$ by a word $uqv\Box$, if the contents of the non-blank part of the tape is $uv$, and $\M$ is reading the first symbol of $v\Box$ in state $q\in Q$.
 
A semi-Thue system $S_{\M}=(\Lambda,R_S)$ imitating the computation of a fixed deterministic TM $\M$ is constructed using the following ideas originally given in \cite{HuLa-78}: $\Lambda=Q\cup\Si\cup\Ga\cup \{L,R\}$, where $L$ and $R$ are end markers. Indeed, the initial configuration   $q_0w$ of $\M$ corresponds to a word $Lq_0wR\in \Lambda^+$  
and the rules of $R_S$ are implied by the transition function $\delta$ so that, for example,   
$$
(aqb, acp)\in R_S \text{ if } \delta(q,b)=(p,c,R),   
$$
and similarly for the other types of transitions. Now it is straightforward to see that a TM $\M$ halts on input $w$ in the configuration $uHv$ for some words $u$ and $v$ if and only if $Lq_0wR\derive{}_{S_{\M}}^* LuHvR$. 
Since the halting problem of TM's on empty tape is undecidable, we may assume in the above that $w=\Box$. 

We obtain a simple proof for undecidability of the the word problem, see \cite{HuLa-78}, by adding  letter-by-letter cancellation rules such that $LuHvR\derive{}^* LHR$ to the semi-Thue system $S_{\M}$. Furthermore, by adding a  special rule 
\begin{equation}\label{eka:sperule}
(LHR,Lq_0wR)
\end{equation}
we have a semi-Thue system with undecidable circular word problem. The semi-Thue system constructed is $Q$-\emph{deterministic} meaning that in all  rules $(u,v)$, both $u$ and $v$ contain exactly one symbol from set $Q$.  

In order to prove that the conjugate-PCP is undecidable, we need to modify the above construction a bit. First of all, we take another copy of the  alphabet $\Lambda$, say $\ov{\Lambda}=\{\ov{a} \mid a\in \Lambda\}$ and add also overlined copies of all rules except the rule~\eqref{eka:sperule} to the system. The special rule~\eqref{eka:sperule} is replaced by two new rules, 
\begin{equation}\label{toka:sperule}
(LHR,\ov{Lq_0\Box R}) \text{ and }(\ov{LHR},Lq_0\Box R).
\end{equation}
Now the circular derivation 
$$
Lq_0\Box R\derive{}_{S_{\M}}^* LuHvR \derive{}^* LHR\derive{} Lq_0\Box R
$$
is transformed into circular derivation 
$$
Lq_0\Box R\derive{}^* LuHvR \derive{}^* LHR\derive{} \ov{Lq_0\Box R}\derive{}^*
\ov{LuHvR} \derive{}^* \ov{LHR}\derive{} Lq_0\Box R
$$
in our new system.

Finally, we simplify the alphabet $\Lambda$ (and $\ov{\Lambda})$. Indeed, we encode injectively the letters in  $\Lambda \setminus ({Q} \cup \{L,R\})$ into $\{a,b\}^+$, and denote the new alphabets $A=\{a,b,L,R\}$ and $B=Q$. We have now constructed a  
semi-Thue system $T_{\M} = (\Sigma, \mathcal{R})$ with the following properties:

\begin{enumerate}

\item $\Sigma = A \cup \ov{A} \cup B \cup \ov{B}$ with pairwise disjoint alphabets $A, \ov{A}, B, \ov{B}$. Notably $A= \lbrace a,b,L,R \rbrace$ where $L,R$ are markers for the left and right border of the word, respectively.

\item $T_{\M}$ is $(B \cup \ov{B})$-deterministic in the following way:

\begin{enumerate}[(i)]

\item $\mathcal{R} \subseteq (A^*BA^* \times A^*BA^*) \cup (\ov{A^*BA^*} \times \ov{A^*BA^*}) \cup (A^*BA^* \times \ov{A^*BA^*}) \cup (\ov{A^*BA^*} \times A^*BA^*)$.

\item If $t_i$ is a rule in $\mathcal{R}$ where none of the symbols are overlined, then the corresponding overlined rule $\ov{t_i}$, where all symbols are overlined is also in $\mathcal{R}$, and vice versa.

\item For all words $w \in (A \cup \ov{A})^*(B \cup \ov{B})(A \cup \ov{A})^*$, if there is a rule in $\mathcal{R}$ giving $w \derive{}_T w'$ then the rule is unique.

\item There is a single rule from $A^*BA^* \times \ov{A^*BA^*}$ and a single rule from $\ov{A^*BA^*} \times A^*BA^*$, moreover these rules are such that they re-write everything between the markers $L$ and $R$, namely if there are rules giving $u \derive{}_{T_{\M}} \overline{w_0}$ and $\ov{u} \derive{}_{T_{\M}} w_0$ for a $u \in A^* B A^* $ then the rules are $(u,\ov{w_0})$ and $(\ov{u},w_0)$, respectively. These rules are the rules in~\eqref{toka:sperule} coded into $\Sigma$.

\end{enumerate}

\item\label{circular} $T_{\M}$ has an undecidable circular word problem. In particular it is undecidable whether $T$ has a circular derivation $w_0 \derive{}_{T_{\M}} ^* w_0$ where $w_0 \in A^*BA^*$ is the word appearing in the rules of 2(iv). Note that $w_0$ and  $u$ in the case 2(iv) are fixed words from the construction of the semi-Thue system $T_{\M}$ for a particular Turing machine $M$, and $w_0 \neq u$.

\end{enumerate}

The special $(B \cup \ov{B})$-determinism of $T_{\M}$ can be interpreted as derivations being in two different phases: the normal phase and the overlined phase. Transitioning between phases happens via the unique rules from 2(iv). It is straightforward to see that all derivations do not go through phase changes and that the phase is changed more than once if and only if $T$ has a circular derivation. The system considered is now fixed from the context and we write the derivations omitting the index $T$ simply as $\derive{}$.

We now add a few additional rules to $T_{\M}$: we remove the unique rule $(u,\ov{w_0})$ and replace it with one extra step by introducing rules $(u,s)$ and $(s,\ov{w_0})$ where $s$ is a new symbol for the intermediate step. The corresponding overlined rules $(\ov{u},\ov{s})$ and $(\ov{s},w_0)$ are added also to replace the rule $(\ov{u},w_0)$. These new rules are needed in identifying the border between words $u$ and $v$, and adding them has no effect on the behaviour of $T_{\M}$.

By the case \ref{circular} of the properties of $T_{\M}$ we have the following lemma.

\begin{lemma}
\label{wordproblem}
Assume that the semi-Thue system $T_{\M}$ is  constructed as in the the above. Then $T_{\M}$ has an undecidable individual circular word problem for the word $w_0$.
\end{lemma}

We now reduce the individual circular word problem of the system $T_{\M}$ to the conjugate-PCP.

Let $\mathcal{R} = \lbrace t_0, t_1, \ldots, t_{h-1}, t_h \rbrace$, where the rules are $t_i = (u_i,v_i)$. We denote by $l_x$ and $r_x$ the left and right desynchronizing morphisms defined by 
$$
l_x(a)=xa, \qquad r_x(a) = ax
$$
for all words $x$.
In the following we consider the elements of $\mathcal{R}$ as letters. Denote by $A_j$ the alphabet $A$ where letters are given subscripts $j=1$ and $2$, respectively. Define the morphisms $h,g : (A_1 \cup A_2 \cup \ov{A_1} \cup \ov{A_2} \cup \lbrace \#,\overline{\#}, I \rbrace \cup \mathcal{R})^* \rightarrow \lbrace a,b,d,e,f,\#,\$,\pounds \rbrace^*$ according to the following table:

\begin{center}
\begin{tabular}{m{2em} | C{10em} | C{10em}  C{10em} }
\setlength\extrarowheight{2.5pt}
& $h$ & $g$ \\
\cline{1-3}
$I$ & $\$ l_{d^2}(w_0 \#)d$ & $\pounds ee,$\\ 
$x_1$ & $dxd$ & $xee,$ & $x \in \lbrace a,b \rbrace$\\ 
$x_2$ & $ddx$ & $xee,$ & $x \in \lbrace a,b \rbrace$\\ 
$t_i$ & $d^{-1}l_{d^2}(v_i)$ & $r_{e^2}(u_i),$ & $t_i \not\in \lbrace t_{h-1},t_h \rbrace $ \\ 
$t_{h-1}$ & $ds f f $ & $r_{e^2}(u \#) $ \\ 
$t_h$ & $f \$ \pounds l_{e^2}(w_0 \#) e e$ & $s fff \pounds \$ d d$ \\ 
$\#$ & $dd\#d$ & $\#ee$ \\ 
$\overline{x_1}$ & $xee$ & $xdd,$ & $\ov{x} \in \lbrace \ov{a},\ov{b} \rbrace$ \\ 
$\overline{x_2}$ & $exe$ & $xdd,$ & $\ov{x} \in \lbrace \ov{a},\ov{b} \rbrace$\\ 
$\overline{t_i}$ & $e^{-2}l_{e^2}(v_i)e$ & $r_{d^2}(u_i),$ & $\ov{t_i} \not\in \lbrace \ov{t_{h-1}}, \ov{t_h} \rbrace$ \\ 
$\overline{t_{h-1}}$ & $sf $ & $r_{d^2}(u \#) $ \\
$\overline{t_h}$ & $f f \pounds $ & $s f f f \$ $ \\
$\overline{\#}$ & $e\# ee$ & $\# dd$ \\

\end{tabular}
\end{center}

\noindent
Here the re-writing rules are of the form $t_i = (u_i,v_i),$ for $u_i,v_i$. The following rules play important roles: \\
$t_{h-1}=(u,s),$ where $u$ is the unique word such that $(u,\ov{w_0}) \in R$, and \\
$t_{h} = (s, \ov{w_0})$.
\\

We begin by examining the forms of the images of $h$ and $g$. The morphisms are modified from the ones in \cite{bi-inf} with slight alterations made such that it is possible to have (finite) solutions to the instance of the conjugate-PCP with easily identifiable borders between the factors $u$ and $v$ using special symbols $\$$ and $\pounds$. The symbols $d,e$ and $f$ function as desynchronizing symbols. The desynchronizing symbols $d$ and $e$ make sure that in the solution $w$ the factors that will represent the configurations of the semi-Thue system $T_{\M}$ are of correct form, that is of the form where determinism is kept intact. This follows from the forms of $h$ and $g$: under $g$ all images are desynchronized by either $e^2$ (non-overlined letters) or $d^2$ (overlined letters). To have similarly desynchronized factors in the image under $h$ we note that in the pre-image the words between two $\#$-symbols (similarly for overlined symbols $\ov{\#}$) are of the form $\alpha t \beta$ where $\alpha \in \lbrace a_1,b_1 \rbrace$, $\beta \in \lbrace a_2,b_2 \rbrace$ and $t \in \mathcal{R}$ (with end markers $L$ and $R$ omitted from $\alpha$ and $\beta$). The symbol $f$ is not really used in desynchronizing but makes sure that the change between phases is carried out correctly. 

The following lemma is useful in our proof:

\begin{lemma}
\label{conjugates}
The words $h(w)$ and $g(w)$ are conjugates if and only if $h(w_1)$ and $g(w_2)$ are conjugates for all conjugates $w_1$ and $w_2$ of $w$.
\end{lemma}
\begin{proof}
If $h(w_1)$ and $g(w_2)$ are conjugates for all conjugates $w_1$ and $w_2$ of $w$ then of course $h(w)$ and $g(w)$ are conjugates.

Assume then that $h(w)$ and $g(w)$ are conjugates and let $w_1$ and $w_2$ be conjugates of $w$. There are then suffixes
$x$ and $y$ of $w$
such that $w_1 = x w x^{-1}$ and $w_2= y w y^{-1}$. Denote $w x ^{-1} = w'$ and $w y^{-1}=w''$. Now $h(w_1)=h(xw')=h(x)h(w')$ is a conjugate of $h(w')h(x)=h(w'x)=h(w)$ and $g(w_2)=g(yw'')=g(y)g(w'')$ is a conjugate of $g(w'')g(y)=g(w''y)=g(w)$. By our assumption also $h(w_1)$ and $g(w_2)$ are conjugates.
\end{proof}

Next we will show that a circular derivation beginning from a fixed word $w_0$ exists in $T_{\M}$ if and only if there is a solution to the conjugate-PCP instance $(h,g)$. We prove the claim in the following two lemmata.

\begin{lemma}
\label{main1}
If there is a circular derivation in $T_{\M}$ beginning from $w_0$, then there exists a non-empty word $w$ such that $h(w)\sim_2 g(w)$.
\end{lemma}
\begin{proof}
Assume that a circular derivation exists. The derivation is then of the form $w_0 = \alpha_1 u_1 \beta_1 \rightarrow \alpha_1 v_1 \beta_1 = \alpha_2 u_2 \beta_2 \rightarrow \cdots \rightarrow u \rightarrow s \rightarrow \ov{w_0} = \ov{\alpha_1 u_1 \beta_1} \rightarrow \cdots \ov{u} \rightarrow \ov{s} \rightarrow w_0$, where $s$ and $u$ as defined earlier for $T_{\M}$. This derivation can be coded into a word

$$w=Iw_1 \# w_2 \# w_3 \# \cdots \# t_{h-1}t_h \ov{w_1 \# w_2 \# w_3 \# \cdots \# t_{h-1}t_h},$$
where $w_i=\alpha_i t_i \beta_i$ for each $i$, where  $t_i=(u_i,v_i)$ is the unique rewriting rule used in each derivation step. The rules $t_{h-1}$ and $t_h$ appear right before transition to overlined part of the  derivation as they correspond to the final and intermediate steps before the transition. Let us consider the images of $w$ under the morphisms $h$ and $g$ defined in the above:
$$
h(w) = \$ l_{d^2}(w_0 \# \alpha_1 v_1 \beta_1 \# \alpha_2 v_2 \beta_2 \# \cdots \# s)fff \$ \pounds l_{e^2}(w_0 \# \alpha_1 v_1 \beta_1 \cdots \#s)fff \pounds)
$$
and
$$
g(w) = r_{e^2}(\pounds \alpha_1 u_1 \beta_1 \# \alpha_2 u_2 \beta_2 \# \cdots \# u \#)sfff \pounds r_{d^2}(\$ \alpha_1 u_1 \beta_1 \cdots \# u \#) sfff \$).
$$
These images are indeed very similar. They match at all positions that do not contain a desynchronizing symbol ($d$ or $e$) or a special symbol ($\$$ or $\pounds$). Thus, if we erase all of these non-matching symbols we would have equality (and of the form $q^2$ for a word $q$). Also the non-matching symbols are such that $d$ is always matched with $e$ and $\$$ is always matched with $\pounds$. It is clear that the factors in both $h(w)$ and $g(w)$ beginning and ending with the same special symbol are the same, that is, the factor of the form $\$ \cdots \$ $ and the factor of the form $ \pounds \cdots \pounds$ both images are equal. It follows that $h(w)\sim_2 g(w)$, which proves our claim.
\end{proof}

\begin{lemma}
\label{main2}
If there exists a non-empty word $w$ such that $h(w)\sim_2 g(w)$, then there is a circular derivation in $T_{\M}$ beginning from $w_0$.
\end{lemma}
\begin{proof}

Firstly we show that the factor $f^3$ must appear in $h(w)$ and hence $t_{h-1}t_h$ or $\ov{t_{h-1}t_h}$ has to be a factor in $w$. Assume on the contrary: there is no factor $f^3$ in $h(w)$.

From the construction of $g$ we know that also $h(w)$ must be desynchronized so that between the letters there is either a factor $d^2$ or $e^2$. Conjugation of $g(w)$ does not break this property except possibly in the beginning and the end of $h(w)$ ($h(w)$ could start and end in a single desynchronizing symbol).

Take now the first letter $c$ of $w$. We can assume that it is a non-overlined letter as the considerations are similar for the overlined case. The letter $c$ cannot be $t_{h-1}$ as it would have to be followed by $t_h$: $f^2$ does not appear as a factor under $g$ without $f^3$, and $t_{h-1}\ov{t_h}$ produces $f^4$, which is uncoverable by $g$. From the construction of $h$ we see that the letters following $c$ must also be non-overlined, otherwise the desynchronization would be broken. Thus the desynchronizing symbol is the same for all the following letters. But as we can see from the form of the morphisms $h$ and $g$, we have a different desynchronizing symbols under $g$ for $c$ and its successors. It is clear that $h(g)$ must contain both $d$ and $e$ and so $w$ must have both non-overlined and overlined letters. If there is a change in the desynchronizing symbol in $h(w)$ then it contradicts the form of the images under $g$. Hence we must have the factor $t_{h-1}t_h$ in $w$ to make the transition without breaking the desynchronization.

The images of the factor $t_{h-1}t_h$ are
$$
h(t_{h-1} t_{h}) = d s fff \$ \pounds l_{e^2}(w_0 \#)ee
$$
and
$$
g(t_{h-1} t_{h}) = r_{e^2}(u\#)s fff \$ \pounds dd.
$$

As we can see the desynchronizing symbols do not match. Hence we also must have the overlined copy of this factor in $w$, that is a factor $\ov{t_{h-1}t_h}I$, the images of which are (the letter $I$ is a forced continuation to the overlined factor to account for the special symbols $\$$ and $\pounds$):

$$
h(\ov{t_{h-1}t_h}I) = s fff \pounds \$  l_{d^2}(w_0 \#)d
$$
and
$$
g(\ov{t_{h-1}t_h}I) = r_{d^2}(u\#)s fff \$ \pounds ee.
$$
One of either of these factors has one swap between the symbols $d$ and $e$. From the above we concluded that we need an even number of these swaps as for every factor $t_{h-1}t_h$ we must also have the factor $\ov{t_{h-1}t_h}I$ and vice versa. It is possible that $h(w)$ ends in the letter $f$. In this case the swap happens "from the end to the beginning", i.e., the prefix of a factor doing the swap is at the end of $w$ and the remaining suffix is at the beginning of $w$. The following proposition shows that we can in fact restrict ourselves to the case where the factors $t_{h-1}t_h$ and $\ov{t_{h-1}t_h}$ are intact, that is, the swap does not happen from the end to the beginning of $h(w)$ as a result of the conjugation between $h(w)$ and $g(w)$. At this point we make an observation.

\begin{observ*}
It may be assumed that the first and last symbols of $h(w)$ are $ \$ $ and $\pounds$.
\end{observ*}

Indeed, if $h(w)$ is not of the desired form then it has $\pounds \$ $ as a factor (by above the symbols from $\ov{t_{h-1}t_h} I$ are in $w$). Images of the letters under $h$ do not have $\pounds \$ $ as a factor so there is a factorization $w = w_1 w_2$ such that $h(w_1)$ ends in $\pounds$ and $h(w_2)$ begins with $\$$ ($w_1$ ends in $\ov{t_h}$ and $w_2$ begins with $I$). By {Lemma \ref{conjugates}}, $h(w)$ and $g(w)$ are conjugates if and only if $h(w_2 w_1)$ and $g(w_2 w_1)$ are, where now $h(w_2 w_1)$ has $\$$ as the first symbol and $\pounds$ as the last symbol.

Now by the observation we may assume that $w$ begins with $I$ and ends with $\ov{t_h}$. From this it also follows that if $h(w)=uv$ and $g(w)=vu$ the word $u$ has $\$ $ as the first and the last symbol and $v$ has $\pounds$ as the first and the last symbol. It follows that $w= I \cdots t_h \cdots \ov{t_h}$, where the border between $u$ and $v$ is in the image $h(t_h)$:
\begin{alignat*}{2}
 h(w) \ = \ &\tikzmark{StartBrace1} \$ l_{d^2}(w_0 \#)d  \quad \cdots  \quad  f\$ \tikzmark{EndBrace1} \tikzmark{StartBrace2}\pounds l_{e^2}(w_0 \#)ee \quad  \cdots \quad &&ff \pounds \tikzmark{EndBrace2} \\
 \\
 g(w) \ = \ & \tikzmark{StartBrace3}\pounds ee  \quad \cdots  \quad s fff \pounds\tikzmark{EndBrace3}\tikzmark{StartBrace4} \$ dd \quad \quad \quad \quad \cdots \quad \quad \quad  sf&&ff \$\tikzmark{EndBrace4} \\
\end{alignat*}
\InsertUnderBrace[draw=black,text=black]{StartBrace1}{EndBrace1}{$u$}
\InsertUnderBrace[draw=black,text=black]{StartBrace2}{EndBrace2}{$v$}
\InsertUnderBrace[draw=black,text=black]{StartBrace3}{EndBrace3}{$v$}
\InsertUnderBrace[draw=black,text=black]{StartBrace4}{EndBrace4}{$u$}

Here the border between $u$ and $v$ need not be in the image of the same instance of $t_h$. Nevertheless we know by above that in the image under $g$ the word $u$ begins with $\$ l_{d^2}(w_0 \#)d$. To get this image as a factor of $g(w)$ we must have $t_h \ov{\alpha_1 t_1 \beta_1 \#}$ in $w$, where $t_1 = (u_1,v_1)$ is the first rewriting rule used and $w_0=\alpha_1 u_1 \beta_1$. Now
$$h(t_h \ov{\alpha_1 t_1 \beta_1 \#}) = f \$ \pounds l_{e^2}(w_0\# \alpha_1 v_1 \beta_1 \#)ee$$
which shows that 
\begin{equation}
\label{first}
I \alpha_1 t_1 \beta_1 \# \alpha_2 t_2 \beta_2 \# \text{ occurs in } w
\end{equation}
where by the $(B \cup \ov{B})$-determinism of $T$ the rule $t_2 \in \mathcal{R}$ is the unique rule and $\alpha_1, \alpha_2 \in L\lbrace a_1,b_1 \rbrace^* \cup \lbrace \varepsilon \rbrace$ and $\beta_1, \beta_2 \in \lbrace a_2,b_2 \rbrace^* R \ \cup \lbrace \varepsilon \rbrace$ are unique words such that $g(\alpha_2 t_2 \beta_2)=r_{e^2}(\alpha_2 u_2 \beta_2) = r_{e^2}(\alpha_1 v_1 \beta_1$).
Again,
$$h(I \alpha_1 t_1 \beta_1 \# \alpha_2 t_2 \beta_2 \#) = \$ l_{d^2}(w_0 \# \alpha_1 v_1 \beta_1 \# \alpha_2 v_2 \beta_2 \#)d$$
which is also a factor of $g(w)$ and implies that
\begin{equation}
\label{second}
t_h \ov{\alpha_1 t_1 \beta_1 \# \alpha_2 t_2 \beta_2 \# \alpha_3 t_3 \beta_3 \#} \text{ occurs in } w
\end{equation}
for a unique $t_3 \in \mathcal{R}$ and $\alpha_3 \in L\lbrace a_1,b_1 \rbrace^* \cup \lbrace \varepsilon \rbrace$,  $\beta_3 \in \lbrace a_2,b_2 \rbrace^* R \cup \lbrace \varepsilon \rbrace$.

We can see that the words given by this procedure beginning with $I$ or $t_h$ (as in \ref{first} and \ref{second}, respectively) contain derivations of the system $T_{\M}$ starting from $w_0$ where configurations are represented as words between $\#$-symbols and consecutive configurations in these words are also consecutive in $T_{\M}$ (as is explained in the beginning of the proof), that is, we get from the former to the latter by a single derivation step.

From the finiteness of $w$ it follows that long enough factors of $w$ of the forms  \ref{first} and \ref{second} represent cyclic computations: the configuration  $s$ is reached eventually and from there we have the rule $(s,w_0)$ which starts a new cycle. We conclude that $T_{\M}$ must have a cyclic computation starting from configuration $w_0$.

\end{proof}

Lemmas \ref{wordproblem}, \ref{main1} and \ref{main2} together yield our main theorem:

\begin{theorem}
The conjugate-PCP is undecidable.
\end{theorem}

This result does not generalize to more complex $(1,n))$-permutations using this same construction by say, adding more desynchronizing symbols and border markers for each element in the permutation. The generalization of the conjugate-PCP would be the $(1,n)$-permutational PCP, stated below:

\begin{prob*}[Image Permutation Post Correspondence Problem]
Given two morphisms $h,g: A^* \rightarrow B^*$, does there exist a word $w \in A^+$ and an $n$-permutation $\sigma$ such that $h(w)=u_1 u_2 \cdots u_n$ and $g(w)=u_{\sigma(1)} u_{\sigma(2)} \cdots u_{\sigma(n)}$ for some words $u_1,\ldots,u_n \in B^*$?
\end{prob*}

The reason that our construction does not work for the general $(1,n)$-case is that allowing more factors to be permuted can force solutions that do not describe TM computations. This is because of special cases for different values of $n$ and $\sigma$, but also by the fact that the permutated factors may be single letters. In fact any solution $w$ that produces Abelian equivalent words $h(w)$ and $g(w)$ also has a permutation that makes one of the words into the other. A "simple" proof using the techniques in this chapter is for now deemed unlikely, and some other approach may prove to be more fruitful. Note that the undecidability of the Image Permutation PCP follows already from proof of Ruohonen for $(m,n)$-permutational PCP in \cite{Ruo2}.  

As a related result we note that the PCP for the instances where one of the morphisms is a permutation of the other are undecidable. Indeed, it was shown by Halava and Harju in \cite{HaHaBullet} that the PCP is undecidable for instances $(h,h\pi)$, where $h: A^* \rightarrow B^*$ is a morphism and $\pi: A^* \rightarrow A^*$ is a permutation. 

\section{Complexity of $\Z$PCP}

In this section, we will consider the  $\Z$PCP defined in the introduction. As mentioned, undecidability of the $\Z$PCP was proved in \cite{bi-inf} using similar techniques than in the previous section for the conjugate-PCP. We begin by reformulating the problem in more details: 

\setcounter{prob}{1}

\begin{prob}[$\Z$PCP]
Let $A$ be a finite alphabet. Given a finite set  of pairs of words over $A$, say $(u_1,v_1),(u_2,v_2),  \dots,$ $(u_n,v_n)$,  does there exist a bi-infinite sequence 
$\ldots i_{-k} \ldots i_{-1},  i_0, i_1,\dots,i_k,\dots$ of the indices such that 
\[
\cdots u_{i_{-k}}\cdots u_{i_{-1}}u_{i_0}u_{i_1} \cdots u_{i_k} \cdots  = \cdots v_{i_{-k}}\cdots v_{i_{-1}}v_{i_0}v_{i_1} \cdots v_{i_k} \cdots  \, ?
\]
\end{prob}

The equality of two bi-infinite words is an equivalence of the sequences of symbols modulo a finite shift $s \in \Z$ in the positions of the sequences. 

An instance of the $\Z$PCP is given by a finite set  of pairs of words over $A$ (which can be coded by an integer via a recursive coding) and a solution to this instance is a 
bi-infinite sequence $( i_k)_{k\in\Z}  \in \{1, 2, \ldots , n\}^\Z$. 

We shall need in the sequel the notion of a Turing machine reading infinite words. We now recall these notions. 
\medskip

The {\it first infinite ordinal} is $\om$.
 An $\om$-{\it word} over an alphabet $\Si$ is an $\om$-sequence $a_1a_2a_3 \cdots$, where for all 
integers $ i\geq 1$, ~
$a_i \in\Si$.     
 The {\it set of } $\om$-{\it words} over  the alphabet $\Si$ is denoted by $\Si^\om$.
An  $\om$-{\it language} over an alphabet $\Si$ is a subset of  $\Si^\om$.  For an $\om$-word $\sigma=a_1a_2a_3\cdots$, we denote the prefix $a_1\cdots a_n$ by $\sigma[n]$.

As in the previous section, assume that a {Turing 
machine} $\mathcal{M}$ is of the form $\mathcal{M}=(Q, \Si, \Ga, \delta, q_0, F)$, where $F\subseteq Q$ is the set of {accepting states}. Turing machines reading of infinite words have considered in~\cite{CG78b,Staiger97}.
A Turing machine $\mathcal{M}$ accepts a word 
$\sigma\in \Si^\om$ with the $2$\emph{-acceptance  condition} iff there is an infinite run of $\mathcal{M}$ on input $\sigma$ 
visiting infinitely often states from $F$. The $2$-acceptance condition is also now known as the B\"uchi acceptance condition. 
On the other hand,  a Turing machine $\mathcal{M}$ accepts a word 
$\sigma\in \Si^\om$ with  $2'$-\emph{acceptance  condition} iff there is an infinite run of $\mathcal{M}$ on $\sigma$ 
visiting only finitely often the accepting states in  $F$. The $2'$-acceptance condition is also now known as the co-B\"uchi acceptance condition. 

We require in this article that an accepting  run should be infinite on the input $\sigma\in \Si^\omega$, as in \cite{Staiger97}, and not that it is complete (i.e. we do not require that all the cells of the right-infinite tape of the Turing machine are visited nor that all letters of $\sigma$ are read), or even non-oscillating, as in \cite{CG78b}. We refer the interested reader to \cite{Fin-ambTM} and papers cited in \cite{Staiger97,Fin-ambTM}  for a comparison between these modes of acceptance of infinite words by Turing machines.

We assume the reader to be familiar with the arithmetical hierarchy on subsets of  $\mathbb{N}$, as a general reference we give \cite{rog,Odifreddi1}. We now recall the definition of  the arithmetical   hierarchy   on subsets of 
$\Si^\om$ for a finite alphabet $\Si$, see \cite{Staiger97}. 
 An $\om$-language $L\subseteq \Si^\om$  belongs to the class 
$\Sigma^0_n$ iff there exists a recursive relation 
$R_L\subseteq \mathbb{N}^{n-1}\times \Si^\star$  such that
\[
L = \{\sigma \in \Si^\om \mid \exists x_1 Q_2 x_2\ldots Q_n x_n  \quad (x_1,\ldots , x_{n-1}, 
\sigma[x_n+1])\in R_L \}, 
\]
\noindent where $Q_i$  for $i=2,\dots, n$ is one of the quantifiers $\fa$ or $\exists$ 
(not necessarily in an alternating order). An $\om$-language  $L\subseteq \Si^\om$  belongs to the class 
$\Pi^0_n$ iff its complement $\Si^\om - L$  belongs to the class 
$\Sigma^0_n$.  
The inclusion relations that hold  between the classes $\Sigma^0_n$ and $\Pi^0_n$ are 
the same as for the corresponding classes of the Borel hierarchy. The classes $\Sigma^0_n$ and $\Pi^0_n$ are strictly included in the respective classes 
${\bf \Sigma}_n^0$ and ${\bf \Pi}_n^0$ of the Borel hierarchy.

An important result is that the modes of acceptance of  $\om$-languages by deterministic  Turing machines are connected  to the classes of the  arithmetical   hierarchy.  In particular,  an $\om$-language is in the   arithmetical  class $\Pi^0_2$ (respectively, $\Sigma^0_2$) if and only if it is accepted by a {\it deterministic } Turing machine with $2$-acceptance condition, i.e. B\"uchi acceptance condition (respectively, with $2'$-acceptance condition, i.e. co-B\"uchi acceptance condition), see Corollary 2.3   in \cite{Staiger97}. \medskip

We now state the main result of this section. 

\begin{theorem}
The bi-infinite PCP is in the class $\Sigma^0_2 \setminus \Pi^0_1$. 
\end{theorem}

\noindent In the above statement the arithmetical classes refer to classes of sets of integers. Indeed, this means that the set of instances of the  bi-infinite PCP having a solution can be recursively coded by a set of integers in the class $\Sigma^0_2 \setminus \Pi^0_1$.

\begin{proof}
We firstly show that the $\Z$PCP is not in the class $\Pi^0_1$. This is actually a direct consequence of the proof of the undecidability of the $\Z$PCP in \cite{bi-inf}. Indeed, the proof shows that there exists  a reduction of the halting problem for Turing machines to the  $\Z$PCP. On the other hand, it is well known that the  halting problem for Turing machines is 
 $\Sigma^0_1$-complete, hence the $\Z$PCP is  $\Sigma^0_1$-hard and, in particular, it is not in the class $\Pi^0_1$. 

Secondly, we prove that the  $\Z$PCP is in the class $\Sigma^0_2$.  Let us consider an instance of the $\Z$PCP given by a finite set  of pairs of words over $A$, where $A$ is a finite alphabet, say $Ins=\{(u_1,v_1),(u_2,v_2),  \dots,$ $(u_n,v_n)\}$.  

In the first step, we are going to associate to this instance a deterministic Turing machine with  co-B\"uchi acceptance condition which accepts exactly the codes of the solutions to the instance $Ins$. For this we encode a bi-infinite sequence of integers in $\{1, 2, \ldots, n\}$ 
 \begin{equation}\label{possol1}
\ldots  i_ {-k} \ldots  i_ {-2}  i_ {-1} i_0 i_1 i_2 \ldots   i_k \ldots
\end{equation}
of a possible solution of the instance $Ins$ for the $\Z$PCP into a pair of infinite sequences
$$ 
i_0 i_1 i_2 \ldots  i_k  \ldots\quad\text{ and }\quad i_0  i_ {-1} i_ {-2} \ldots  i_ {-k}\ldots.
$$
Now both of these sequences are  infinite words over the alphabet $\{1, 2, \ldots, n\}$, so that we can code the bi-infinite sequence  in~\eqref{possol1} into an 
$\omega$-word over the finite alphabet 
 $\{1, 2, \ldots, n\}\times \{1, 2, \ldots, n\}$ so that $(i_j)_{j\in \Z}$ corresponds to   
\begin{equation}\label{possol2}
    (i_0,i_0)(i_1, i_ {-1})(i_2, i_ {-2}) \ldots ( i_k, i_ {-k})\ldots .
\end{equation}

Next we show that the set of (codes of) solutions of the instance $Ins$ of the $\Z$PCP are accepted by a deterministic Turing machine $\mathcal{M}$ with  co-B\"uchi acceptance condition, (with set of final states $F$ for the  co-B\"uchi acceptance condition). We informally explain the behaviour of this Turing machine:

The TM $\mathcal{M}$ works with integers $s$ denoting the shift in the images. Indeed, first the shift $s=0$. For an input of the form~\eqref{possol2}, denote by $u= 
\ldots u_{ i_ {-k}}\ldots  u_{ i_ {-2}} u_{  i_ {-1}} u_ { i_0} u_{ i_1 } u_{i_2}\ldots u_{ i_k } \ldots$ and 
 $v= \ldots v_{ i_ {-k}}\ldots  v_{ i_ {-2}} v_{  i_ {-1}} v_ { i_0} v_{ i_1 } v_{i_2}\ldots v_{ i_k } \ldots$
There is a (possibly infinite test) we call TEST: 
\medskip

TEST: For $m=0,1\dots $, check that $u(m)=v(m+s)$ and $u(-m)=v(-m+s)$.  

\medskip
If the TEST fails and for some $m$ one of the equations does not hold (that is, $\mathcal{M}$ found an error and the sequence is not a solution with the shift $s$), then  
$\mathcal{M}$ enters in some state in $F$, and sets $s:= -s$ if the TEST is done odd number of times and $s:=|s|+1$ it is done even number of times, and does the TEST for that new $s$. 

It is rather obvious that if $\mathcal{M}$ visits states of $F$ only finitely many times, the sequence \eqref{possol2} codes a solution of the instance $Ins$ as for some shift $s$ $\mathcal{M}$ found no error, that is,
$u=v$ modulo some shift $s \in \Z$.
Indeed, then the TM  $\mathcal{M}$ accepts the coding ~\eqref{possol2}  of the bi-infinite sequence of integers \eqref{possol1} with   co-B\"uchi acceptance condition. Conversely, if the sequence \eqref{possol2} codes a solution to the instance  $Ins$ of  $\Z$PCP, then  \eqref{possol2} is accepted by the TM $\mathcal{M}$ with co-B\"uchi acceptance condition. 

We do not go into the details of defining $\mathcal{M}$ but note that it is an easy exercise to construct such a deterministic Turing machine with  co-B\"uchi acceptance condition from the instance $Ins$. 

Now the set of infinite words accepted by such a deterministic Turing machine  with  co-B\"uchi acceptance condition is known to be an effective $\Sigma^0_2$ set. Moreover,
 Cenzer and Remmel proved in  \cite[Theorem 4.1.(iii)]{CenzerRemmel03}
that the non-emptiness problem for such  effective $\Sigma^0_2$ sets is in the class  $\Sigma_2^0$.  
Thus the problem to determine whether a given instance $Ins$ of the $\Z$PCP has a solution is in the class $\Sigma_2^0$. 
\end{proof}

The next goal of this study of the complexity of the $\Z$PCP would be to determine its exact complexity.  In particular, is it located at the second level of the arithmetical hierarchy? Is it $\Sigma_2^0$-complete?  We leave these questions  as an open problem for further study.


\end{document}